\documentclass[runningheads]{llncs}
\usepackage{graphicx}
\usepackage{booktabs} 
\usepackage[ruled]{algorithm2e} 

\SetAlFnt{\small}
\SetAlCapFnt{\small}
\SetAlCapNameFnt{\small}
\SetAlCapHSkip{0pt}
\IncMargin{-\parindent}

\usepackage{optidef}
\usepackage{bbm}
\usepackage{tikz}
\usetikzlibrary{decorations.pathmorphing}
\usepackage{algorithmic}
\usepackage{graphicx}
\usepackage{textcomp}
\usepackage{xcolor}
\usepackage{microtype}
\usepackage{bm}
\usepackage{url}
\usepackage{microtype}

\newcommand{\remove}[1]{}
\newcommand{\h}{\hspace*{0.2in}}

\begin{document}
\title{Improved Paths to Stability for the Stable Marriage Problem
\thanks{Supported by NSF CNS-1812349, CNS-1563544, and the Cullen Trust for Higher Education Endowed Professorship}
}
%
%

\author{Vijay K. Garg\orcidID{0000-0002-5797-4389} \and Changyong Hu\orcidID{0000-0003-4074-0313}}
%
%
\institute{
The University of Texas at Austin,\\
\email{\{garg|colinhu9\}@ece.utexas.edu}
}
\maketitle   
\begin{abstract}
The stable marriage problem has wide applications in distributed computing such as the placement of virtual
machines in a distributed system.
The stable marriage problem requires one to find a marriage with no blocking pair.
Given a matching that is not stable, Roth and Vande Vate have shown that there exists a sequence of matchings that
leads to a stable matching in which each successive matching is obtained by satisfying a blocking pair.
The sequence produced by Roth and Vande Vate's algorithm is of length $O(n^3)$ where $n$ is the number of men (and women).
In this paper, we present an algorithm that achieves stability in a sequence of matchings of length $O(n^2)$. We also give an efficient algorithm to find the stable matching closest to the given initial matching under an appropriate distance function between matchings.
\end{abstract}

\section{Introduction}
The Stable Matching Problem \cite{gale1962college} has wide applications in distributed computing such as 
the placement of virtual machines in a distributed system \cite{maggs2015algorithmic} or
the placement of files in a distributed system.
It has applications in many other numerous fields such as economics and resource allocation with multiple books 
and survey articles \cite{gusfield1989stable,knuth1997stable,roth1992two,iwama2008survey,david2013algorithmics}.
In the standard version of the problem, there are $n$ men and $n$ women each with their totally ordered preference list. 
The goal is to find 
a matching between men and women such that there is no blocking pair, i.e., there is no pair of a woman and a man such that they are not married to each other but
prefer each other over their partners. 
The standard Gale-Shapley (GS) algorithm produces such a matching starting from an empty matching with the deferred acceptance proposal algorithm that takes
$O(n^2)$ proposals. The algorithm produces the man-optimal stable matching.

In many applications, it is useful to consider the initial state of the system as an arbitrary assignment of men to women and then to find
a path to a stable matching. For example, suppose that we consider a system in which there are more women than men and suppose
that every man is matched to a unique woman such that there is no blocking pair. Now, if a new man or a woman
joins the system, it is more natural to start with the initial state as the existing assignment rather than the empty matching.
In particular, if there is some cost associated with breaking up an existing couple, then we may be interested in the paths to stability
that are of short lengths. Hence, this generalization allows one to consider {\em incremental} stable matching algorithms.

As another example, suppose that we have a stable matching. In a dynamic preference mechanism, a woman may change her
list of preferences. The existing matching may not be stable under new preferences of the woman. Again, it is more natural to
start with the existing matching and then to find a path to a stable matching under new preferences. Thus, the generalization allows one
to consider a {\em dynamic} stable matching algorithm in which preferences of a man or a woman may change and the goal is to find
a stable matching under new preferences. 

Given a matching, a natural method to make progress towards a stable matching is as follows.
The man and the woman in the blocking pair
are married and their spouses are divorced. By marrying these divorcees,
we get another matching. 
The reader is referred to the book \cite{david2013algorithmics} for a detailed discussion
of algorithms that go from a matching to a stable matching.
Knuth \cite{knuth1997stable} showed that starting from any matching and iteratively satisfying a blocking pair may lead to 
a cycle. Abeledo and Rothblum \cite{abeledo1995paths} have shown that a cycle exists even if one chooses the best blocking pair to satisfy at each step.
A pair $(p, q')$ is the best blocking pair for $p$ if for any other blocking
pair $(p, q)$ in $M$, $p$ prefers $q'$ to $q$. 
Indeed, it has been shown by Tamura \cite{tamura1993transformation} and independently by Tan and Su \cite{tan1995divorce}, that there are matchings for which it is not possible to
reach a stable marriage by marrying off divorcees.  
However, if the divorcees are allowed to remain single, then one can achieve stability. The
Roth and Vande Vate (RVV) mechanism \cite{roth1990random} is the most well known method to determine a path to stability. Their algorithm
introduces agents (men or women) incrementally and let them iteratively reach a stable matching.
Given any matching $M_0$, the RVV mechanism produces a sequence of matchings $M_0, M_1, \ldots, M_t$ such that
$M_t$ is stable matching and for each $k$ ($1 \leq k \leq t$), $M_k$ is obtained from $M_{k-1}$ by satisfying a blocking pair.
The value of $t$ is at most $2n^3$ (assuming that the number of acceptable pairs
is $n^2$).

In this paper, we analyze the path to stability from the perspective of traversal in the {\em proposal vector lattice}. Any man-saturating matching corresponds to a unique
proposal vector but when the matching is incomplete there may be multiple proposal vectors corresponding to it. Working with proposal vectors instead of matching
allows us to generate shorter sequences to a  stable matching. In particular, we show that given any proposal vector $G_0$, there exists a sequence
of proposal vectors $G_0, G_1, \ldots, G_t$ such that $G_t$ corresponds to a stable matching and for
each  $k$ ($1 \leq k \leq t$), $G_k$ is obtained from $G_{k-1}$ by either increasing the choice number for one man (thereby worsening his match) or
decreasing the choice number for one man (thereby improving his match). The value of $t$ is at most $2m^2$ where $m$ is the number of men.
Our result can also be phrased in terms of matching as follows.
Given any matching $M_0$, there exists a sequence
of matching $M_0, M_1, \ldots, M_t$ such that $M_t$ corresponds to a stable matching and for
each  $k$ ($1 \leq k \leq t$), $M_k$ is obtained from $M_{k-1}$ by either 
(1) marrying a man whose current partner is in a blocking pair (or, if he is single) to another woman who is agreeable to his proposal, or
(2) by marrying a woman to her best blocking pair partner.
 The value of $t$ is at most $2m^2$ where $m$ is the number of men.
 Thus, this sequence is shorter than the RVV sequence by a factor of $m$.

We propose four algorithms in this paper for achieving stability (see Fig. \ref{fig:lattice}). The first algorithm $\alpha$ is a generalization of the GS algorithm
to find the man-optimal marriage. The GS algorithm starts with the matching that results when all men propose to their top choice. It then
determines the man-optimal stable marriage in $O(n^2)$ moves. What if instead of the top choices, men propose to
any arbitrary vector of women? In such a scenario, a woman cannot accept the first proposal she receives (as in the GS algorithm),
because that may result in an unstable matching. Algorithm $\alpha$ gives the rules for advancing from an arbitrary proposal vector to
end up in a stable marriage (whenever possible). 
 Given any initial matching $M_0$, algorithm $\alpha$ produces
a sequence of matchings ending in a stable marriage $M_t$ such that the matching only improves from the perspective of women
and gets only worse from the perspective of men. This sequence is of length $O(n^2)$. Since there may not exist any stable matching
that is based only on improving from the women's perspective, algorithm $\alpha$ may return null in these cases (for example, when the 
initial matching $M_0$ assigns some man a partner who is ranked lower than in the man-pessimal matching). 
The set of stable matchings can be viewed  as a sublattice of the lattice of all proposal vectors and the algorithm
$\alpha$ can be viewed as upward traversal in this lattice from any arbitrary proposal vector to a proposal vector that corresponds to a stable matching.

One of the goals of the paper is to find a matching that is not too far from the original matching
(or the initial proposal vector).
Given any proposal vector $I$, the {\em regret} of a man is defined as the rank the woman he is assigned in $I$, i.e.,
 if a man is assigned his $k^{th}$ top choice in $I$ then his regret is $k$. 
Given two proposal vectors $I$ and $M$, we define the distance between 
$I$ and $M$, $dist(I, M)$ as the sum of differences of regrets for all men in $I$ and $M$, 
i.e. the $L_1$ distance between two vectors, $dist(I, M) = \|I - M\|_1$.
Algorithm $\alpha$ guarantees that the stable matching $M_t$ computed has the least distance of all stable matchings that
are better than $I$ from the women's perspective. 

The second algorithm $\beta$ does the downward traversal in the proposal lattice in search of a stable marriage.
Algorithm $\beta$ also takes an arbitrary proposal vector $I$ as the starting point and results in a stable marriage whenever possible.
It improves the matching from the perspective of men. When men and women are equal then such a traversal
can be accomplished by switching the roles of men and women. However, in this paper we assume that
the number of men $m$ may be much smaller than the number of women $w$.  All our algorithms have 
time complexity of $O(m^2 + w)$. Switching the roles of men and women is not feasible
without increasing the complexity of our algorithms. 
Algorithm $\beta$ guarantees that the stable matching $M_t$ computed has the least distance of all stable matchings that
are better than $I$ from the men's perspective. 

The third algorithm $\gamma$ combines a downward traversal with an upward traversal to guarantee that 
irrespective of the initial matching $I$, there always exists a sequence of matchings that results in a stable matching.
This sequence consists of two subsequences each of length $O(m^2)$ giving us the path to stability of length $O(m^2)$, thereby improving
on the RVV mechanism.
Intuitively, the RVV algorithm may traverse the lattice
in the upward direction or downwards direction multiple times. In contrast, our algorithm $\gamma$ traverses the proposal lattice 
once in the downward direction and then in the upward direction ending in a stable matching.
It generates a sequence of proposal vectors that results in a stable matching
with $O(m^2+w)$ time complexity. 

Our last algorithm $\delta$ finds the closest stable matching to the given initial proposal vector.
Algorithm $\delta$ is based on a linear programming formulation of the stable marriage problem by
Rothblum \cite{rothblum1992characterization}. By appropriately defining the objective function to minimize the distance from
the initial proposal vector, we get a polynomial time algorithm to find the closest stable marriage.

Our algorithms are also useful in the context of arriving at more egalitarian matchings than
we get using the Gale-Shapley algorithm. If there are $m$ men, instead of starting with the proposal vector
$(1, 1, \ldots, 1)$, we may start with the proposal vector $(m/2, m/2, \ldots, m/2)$ to find a stable 
vector close to the center of the proposal lattice. Alternatively, we can also start with various
proposal vectors chosen at random and obtain multiple stable matchings. Once we have multiple stable matchings, we can use Teo and Sethuraman's median stable matching theorem \cite{teo1998geometry} to
return the median stable matching.

We note here that the path to a stable matching
from an unstable matching can be of different types. Given any blocking pair, the RVV algorithm is based
on a {\em better response} dynamics. Under these dynamics, any blocking pair $(p, q)$ for a matching $M$ is chosen and they are matched. The partners of $p$ and $q$ in $M$, if any, are unmatched. 
An alternative approach based on {\em best response dynamics} is explored in \cite{ackermann2011uncoordinated}. Here, one side, say the set of women, is considered active and the other side is considered passive. An active agent of a blocking pair $(p,q)$ in $M$
plays the {\em best response} if $p$ is matched to $q'$ such that 
$(p, q')$ is the best blocking pair for $p$. In other words, if there is any other blocking
pair $(p, q)$ in $M$, then $p$ prefers $q'$ to $q$. The paper \cite{ackermann2011uncoordinated} gives an example of a two sided market with three men and women in which {\em best response dynamics} can cycle.
They also propose an algorithm to generate a sequence of $2mw$ best responses from any matching $M$ that leads to a stable matching. Their algorithm has some similarities with our algorithm in that it also consists
of two phases. In the first phase, only matched women can make best response moves whereas in the
second phase all women can play the best response. 
However, a crucial difference from our algorithm is that we are interested in finding
a matching that is close to the original matching (where the distance is defined based on the proposal lattice). The algorithm in \cite{ackermann2011uncoordinated} does not concern itself with the issue
of the distance between matchings. In particular, under the {\em best response dynamics}, their algorithm
has the tendency to get to the woman-optimal marriage irrespective of the initial matching.
In contrast, our algorithms provide guarantees on the matching returned.
For example, Algorithm $\alpha$ returns the proposal vector that has the least distance from $I$, the
initial proposal vector of all the proposal vectors that are bigger than $I$.
 
In summary, the paper makes the following contributions.
\begin{itemize}
    \item It proposes Algorithm $\alpha$ that takes any proposal vector (and therefore any matching) to a stable proposal vector in $O(m^2+w)$ time such that the resulting proposal vector has the least distance from the initial proposal vector of all stable proposal vectors that are greater than initial proposal vector. This algorithm generalizes
    the GS algorithm which assumes the initial proposal vector to be the top choices.
    \item It proposes Algorithm $\beta$ that takes any input proposal vector and generates a stable proposal vector that has the least distance from the initial proposal vector to all stable proposal vectors that are less than or equal to the initial proposal vector.
    This algorithm give a dual of the GS algorithm in which the active agents improve their choices to get a stable proposal vector (if one exists).
    \item It proposes Algorithm $\gamma$ that takes any input proposal vector and always generates a stable proposal vector by combining aspects of algorithms $\alpha$ and $\beta$.
    \item It proposes a polynomial time algorithm $\delta$ (based on linear programming) that takes any input proposal vector and
    generates a stable proposal vector that is closest to the input proposal vector.
\end{itemize}
\begin{figure}
    \centering
    \begin{tikzpicture}[scale=0.90]
    
    \draw[thick] (0,0) -- (6,4.5) -- (2,8) -- (-4,3.5) -- (0,0);
    
    \draw[fill=white] (3,0.8) circle (3pt);
    \node at (4.5,0.8) {unstable vectors};
    \draw[fill=black] (3,1.5) circle (3pt);
    \node at (4.5,1.5) {stable vectors};
    
    \node at (0,-0.3) {$(1,1,\cdots,1)$};
    \draw[fill=white] (0,0) circle (3pt);
    \draw[fill=black] (0.5,1.5) circle (3pt);
    \draw [->,decorate,decoration={snake,amplitude=.4mm,segment length=2mm,post length=1mm}]
(0,0.1) -- (0.5,1.4);
    \node at (-0.6,1.5) {man-optimal};
    \node at (0.55,0.75) {GS};
    
    \draw[fill=black] (-2.5,4) circle (3pt);
    \draw[fill=white] (-1.5,4) circle (3pt);
    \draw [->,decorate,decoration={snake,amplitude=.4mm,segment length=2mm,post length=1mm}]
(-1.6,4) -- (-2.4,4);
    \node at (-2, 4.3) {$\delta$};
    
    \draw[fill=white] (2,5.8) circle (3pt);
    \draw[fill=white] (1.75,5.15 ) circle (3pt);
    \draw[fill=black] (1.5,4.5) circle (3pt);
    \draw [->,decorate,decoration={snake,amplitude=.4mm,segment length=2mm,post length=1mm}]
(2,5.7) -- (1.75,5.25);
    \draw [->,decorate,decoration={snake,amplitude=.4mm,segment length=2mm,post length=1mm}]
(1.75,5.05) -- (1.5,4.6);
    \node at (2.1, 5.15) {$\beta$};
    
    \draw[fill=white] (0,2) circle (3pt);
    \draw[fill=white] (0,3 ) circle (3pt);
    \draw[fill=black] (0,4) circle (3pt);
    \draw [->,decorate,decoration={snake,amplitude=.4mm,segment length=2mm,post length=1mm}]
(0,2.1) -- (0,2.9);
    \draw [->,decorate,decoration={snake,amplitude=.4mm,segment length=2mm,post length=1mm}]
(0,3.1) -- (0,3.9);
    \node at (0.4, 3) {$\alpha$};
    
    \draw[fill=white] (4,5) circle (3pt);
    \draw[fill=white] (3.5,4) circle (3pt);
    \draw[fill=white] (3,3) circle (3pt);
    \draw[fill=black] (2.7,3.7) circle (3pt);
    \draw [->,decorate,decoration={snake,amplitude=.4mm,segment length=2mm,post length=1mm}]
(4,4.9) -- (3.5,4.1);
    \draw [->,decorate,decoration={snake,amplitude=.4mm,segment length=2mm,post length=1mm}]
(3.5,3.9) -- (3,3.1);
    \draw [->,decorate,decoration={snake,amplitude=.4mm,segment length=2mm,post length=1mm}]
(3,3.1) -- (2.7,3.6);
    \node at (3.9, 4) {$\gamma$};

    \draw[fill=white] (2,8) circle (3pt);
    \node at (2,8.3) {$(w,w,\cdots,w)$};
    \draw[fill=black] (1.1,6.5) circle (3pt);
    \node at (2.4,6.5) {woman-optimal};
    
    \end{tikzpicture}
    \caption{A proposal lattice with various traversals. Algorithm GS always starts from the bottom of the lattice and finds the man-optimal vector. Algorithm $\alpha$ starts from any vector $I$ and finds the smallest stable vector which is greater than or equal to $I$ (if any exists). Algorithm $\beta$ finds the largest stable vector which is less than or equal to $I$. Algorithm $\gamma$ always converges to a stable vector. Algorithm $\delta$ finds a stable vector that is closest in Manhattan metric.}
    \label{fig:lattice}
\end{figure}

\vspace*{-0.3in}
\section{Proposal Vector Lattice \label{sec:proposal-vector-lattice}}
We consider stable marriage instances with $m$ men numbered $1,2,\ldots, m$ and $w$ women numbered $1, 2, \ldots w$. 
We assume that the number of women $w$ is at least 
$m$; otherwise, the roles of men and women can be switched.
The variables $mpref$ and $wpref$ specify  the men preferences and the women preferences, respectively.
Thus, $mpref[i][k]=j$ iff woman $j$ is the $k^{th}$ preference for man $i$. 
Fig. \ref{fig:SMP2} shows an instance of the stable matching problem.

We use the notion of a {\em proposal vector} for our algorithms. A (man) proposal vector, $G$, is of dimension $m$, the
number of men. We view any vector $G$ as 
follows: $(G[i]=k)$ if man $i$ has proposed to his $k^{th}$ preference, i.e. the woman given by $mpref[i][k]$.
If $mpref[i][k]$ equals $j$, then $G[i]$ equals $k$ corresponds to man $i$ proposing to woman $j$.
For convenience, let $\rho(G, i)$ denote the woman $mpref[i][G[i]]$.
The vector $(1,1,\ldots, 1)$ corresponds the proposal vector in which every man has proposed to his top choice.
Similarly, $(w,w,\ldots, w)$ corresponds to the vector in which every man has proposed to his last choice.
Our algorithms can also handle the case when the lists are incomplete, i.e., a man prefers staying alone to being matched to
some women. However, for simplicity, we assume complete lists.
It is clear that the set of all proposal vectors forms a distributive lattice under the natural less than order
in which the meet and join are given by the component-wise minimum and the component-wise maximum, respectively.
This lattice has $w^m$ elements.

%

Given any proposal vector, $G$, there is a unique matching defined as follows:
man $i$ and $\rho(G, i)$ are matched in $G$ if the proposal by man $i$ is the best for that woman in $G$.
A man $p$ is unmatched in $G$ if his proposal is not the best proposal for that woman in $G$. A woman $q$
is unmatched in $G$ if she does not receive any proposal in $G$; otherwise, she is matched with the best proposal for her in $G$.


A proposal vector $G$ represents a {\em man-saturating matching} iff no woman receives more than one proposal in $G$.
Formally, $G$ is a man-saturating matching if $\forall i,j: i \neq j: \rho(G, i) \neq \rho(G, j)$.
When the number of men equals the number women, a man-saturating matching is a perfect matching (all men and
women are matched).
When the number of men is less than the number of women, then $G$ is a man-saturating matching if every
man is matched (but some women are unmatched).
We say that a matching $M_1$ (or a marriage) is less than another matching $M_2$
if the proposal vector for $M_1$ is less than that of $M_2$. Thus, the man-optimal
marriage is the {\em least} stable matching in the proposal lattice and woman-optimal
marriage is the {\em greatest} stable matching. 

A proposal vector $G$ may have one or more blocking pairs. 
A pair of man and woman $(p, q)$ is a {\em blocking pair} in $G$ iff
$\rho(G, p)$ is not $q$, man $p$ prefers $q$ to $\rho(G, p)$, and woman $q$ prefers
$p$ to any proposal she receives in $G$. Observe that this definition works even when woman $q$ is unmatched,
i.e. she has not received any proposals in $G$. In this case, woman $q$ prefers $p$ to staying alone, and $p$ prefers $q$ to $\rho(G, p)$.

A proposal vector $G$ is a stable marriage (or a stable proposal vector) iff it is a man-saturating matching and there are no blocking pairs in $G$.
The usual stable matching problem is to determine such a proposal vector given $mpref$ and $wpref$. The problem
that we consider in this paper includes an additional input: the initial proposal vector, $I$. The goal is to traverse the
proposal lattice starting from $I$ to 
find a stable proposal vector $G$.
In this paper, we use two different mechanisms --- upward traversal and downward traversal --- to reach a stable matching proposal vector.

Algorithm $\alpha$ uses upward traversal.
Suppose that $q$ is matched with $p'$ in $G$ and is part of the blocking pair $(p,q)$.
Instead of satisfying the blocking pair $(p,q)$, we {\em move} $p'$ to his next choice in his preference list.
This move makes the proposal vector better from the women's perspective and worse from the men's perspective.
By continuing in this manner if some man makes a proposal to $q$ who is even better than $p$, the blocking pair $(p,q)$ gets eliminated.
If no man better than $p$ ever makes a proposal to $q$, then there is no proposal vector bigger than $G$ that corresponds to a stable matching.

\vspace*{-0.1in}
\begin{small}
\begin{figure}
\begin{tabular}{l | l l l || l l  l | l l l |}
$mpref$ & & &  & & & $wpref$ \\
$m_1$  &   $w_1$ & $w_2$ & $w_3$  &  & & $w_1$ & $m_2$ & $m_1$ & $m_3$ \\
$m_2$   &   $w_2$ & $w_3$ & $w_1$ &  & & $w_2$ & $m_3$ & $m_2$ & $m_1$ \\
$m_3$   &   $w_3$ & $w_1$ & $w_2$ &  & & $w_3$ & $m_1$ & $m_3$ & $m_2$\\
\end{tabular}

\caption{\label{fig:SMP2}  Stable Matching Problem with men preference list ($mpref$) and
women preference list ($wpref$).}
\vspace*{-0.1in}
\end{figure}
\end{small}

Algorithm $\beta$ uses downward traversal in the proposal lattice.
Let $G$ be a proposal vector that is not stable. 
Of all the blocking pairs that $q$ is part of, we choose the best blocking pair
from $q$'s perspective. Let $(p,q)$ be such a blocking pair.
We construct a proposal vector $G'$ that {\em moves} man
$p$ to woman $q$ by changing the proposal of man $p$ from his current proposal to that for woman $q$ and keeping
all other proposals as before. 

Since Algorithm $\alpha$ traverses the lattice upwards, any sequence of proposal vector it generates can be of length at most $m^2$.
Similarly, Algorithm $\beta$ also generates a sequence of length at most $m^2$. Algorithm $\gamma$ combines one downward traversal and
one upward traversal to go from any proposal vector to a stable matching proposal vector in a sequence of length at most $O(m^2)$.

%

%
We now describe Algorithms $\alpha$, $\beta$ and $\gamma$ in detail.

\section{Algorithm $\alpha$}
Given any initial proposal vector $I$, Algorithm $\alpha$, finds a stable matching
$G$ such that $I \leq G$ whenever there exists such a stable matching. 
The initial proposal vector is arbitrary instead of the top choice for each man.
This generalizes the GS algorithm which starts with $I=(1,1,\ldots,1)$.
Observe that the GS algorithm does not work when the starting proposal vector
is arbitrary. The GS algorithm requires men to make proposals and women to accept
the best proposals they have received so far.
If the starting proposal vector is a man-saturating matching but not stable, then each woman gets a unique proposal.
All women would accept the only proposal received, but the resulting marriage would not be stable.

This instability may arise due to two reasons. First,
it may arise when 
 the number of women exceeds the number of men.
If we start with the top choices of all men, then the GS algorithm would still return a man-optimal
stable matching with the excess women unmatched. However, if we start from an arbitrary proposal vector, 
we can end up with all women getting unique proposals but there may exist an unmatched woman who is
preferred by some man over his current match. 

To tackle this problem, we first do a simple check on the initial proposal vector
as given by the following Lemma.
Let $numw(I)$ be the total number of unique women that have been proposed in all vectors that are less than or equal to $I$, i.e., $numw(I) = \#\{j \in [w] : \exists G \leq I, \exists i, \rho(G,i) = j\}$.
\begin{lemma}\label{lem:ubound}
Let $I$ be the initial proposal vector for any stable marriage instance with $m$ men.
There is no stable marriage for any proposal vector
$G \geq I$ whenever $numw(I) > m$.
\end{lemma}
\begin{proof}
Consider any proposal vector $G \geq I$.
Since the total number of men is $m$, there is at least one woman $q$ who has been proposed to in a vector less than $G$ and who does not have any proposal in 
$G$. Suppose that proposal was made by man $p$. Then, man $p$ prefers $q$ to $\rho(G, p)$ and $q$ prefers $p$ to staying alone.
\end{proof}

Hence, in our algorithm we only consider $I$ such that the total number of women proposed until $I$ (in all vectors less than or equal to $I$) is at most $m$.

Instability may arise 
even when the number of men and women are equal.
In Fig. \ref{fig:SMP2}, this situation would arise if we started with $I=(2,2,2)$.
The initial proposal vector may be a perfect matching but not stable. 
A woman $q$ may receive a unique proposal from a man $p$ but she prefers
$p'$ who has made his proposal to $q'$ even though $p'$ prefers $q$ to $q'$. Such a scenario cannot happen
when men propose starting from the top choice and in the decreasing order as in the GS algorithm. However, now the
starting vector is arbitrary and a blocking pair may exist in the man-saturating matching.

To address this problem, we define the notion of a {\em forbidden} man in a proposal vector.
\begin{definition}[forbidden]
A man $i$ is forbidden in $G$ if either he is unmatched in $G$ or matched to a woman in $G$
who is part of a blocking pair.
Formally, the predicate $forbidden(G, i)$ holds if there exists another man $j$ such that
either (1) both $i$ and $j$ have proposed to the same woman in $G$ and that woman prefers $j$, or
(2) $(j, \rho(G, i))$ is a blocking pair in $G$.
\end{definition}


We first show that 
\begin{lemma}\label{lem:forbidden}
Let $G$ be any proposal vector such that $numw(G) \leq m$.
There exists a man $i$ such that $forbidden(G, i)$ iff $G$ is not a stable marriage.
\end{lemma}
\begin{proof}
First suppose that there exists $i$ such that $forbidden(G, i)$. This mean that there exists a man $j$ such that
$j$ has proposed to the same woman and that woman prefers $j$ or $(j, \rho(G, i))$ is a blocking pair in $G$.
If both $i$ and $j$ have proposed to the same woman in $G$, then it is clearly not a matching. If $(j, \rho(G, i))$ is a blocking pair then $G$ is not stable. 

Conversely, assume that $G$ is not a stable marriage. This means that either $G$ is not a man-saturating matching or there is a blocking pair
for $G$. If it is not a man-saturating matching, then there must be some woman who has been proposed by multiple men.
Any man $i$ who is not the most favored in the set of proposals satisfies $forbidden(G, i)$.
If $G$ is a man-saturating matching but not a stable marriage, then there must be a blocking pair $(p,q)$.
If $q$ has been proposed in $G$ by man $i$, then $(p, \rho(G, i))$ is a blocking pair, and therefore
$forbidden(G, i)$ holds.
If $q$ has not been proposed in $G$ then we know there are at least $m+1$ women that are in $numw(G)$ which violates our assumption on $G$.
\end{proof}

\begin{figure}
\fbox{\begin{minipage}  {\textwidth}\sf

{\bf input}: A stable marriage instance, initial proposal vector $I$\\
{\bf output}: smallest stable marriage greater than or equal to $I$ (if one exists)\\
$forbidden(G, i)$ holds if $i$ is unmatched or his partner forms a blocking pair in $G$.\\
\\
(0) If $numw(I) > m$ then return null else $G$ := $I$;\\
(1) {\bf while} there exists a  man $i$ such that $forbidden(G, i)$\\
(2) \h let $q$ be the next woman in the list of man $i$ such that $i$ has the most preferred proposal to $q$,\\
(3) \h if no such choice after $G[i]$ or the number of women proposed including $q$ exceeds $m$ then\\
\h\h \h return  null;  // "no stable matching exists"\\
(4) \h else $G[i]$ := choice that corresponds to woman $q$;\\
(5) {\bf endwhile};\\
(6) return $G$;
\end{minipage}
} 
\caption{Algorithm $\alpha$ that returns the least stable vector greater than or equal to the given proposal vector $I$.\label{fig:alpha}}
\end{figure}

Algorithm $\alpha$ shown in Fig. \ref{fig:alpha} exploits the $forbidden (G, i)$ function to search for the stable marriage
in the proposal lattice. The basic idea is that if a man $i$ is forbidden in the current proposal vector $G$, then
he must go down his preference list until he finds a woman who is either 
unmatched or prefers him to her current match.
The while loop at line (1) iterates until none of the men are forbidden in $G$. If the while loop terminates then $G$ is a stable marriage on account of Lemma \ref{lem:forbidden}. At line (2), man $i$ advances on his preference list until his proposal is the most preferred proposal to the woman among
all proposals that are made to her in any proposal vector less than or equal to $G$. If there is no such proposal, then
there does not exist any $G \geq I$ such that $G$ is stable and in line (3), the algorithm returns null. Otherwise, the man 
makes that proposal at line (4).


For example, consider the initial proposal vector $G =(2,2,2)$ in Fig. \ref{fig:SMP2}. In this proposal vector, we have the matching $\{ (m_1, w_2), (m_2, w_3), (m_3, w_1) \}$.
While this is a man-saturating matching, it is not stable because it has blocking pairs. Consider the blocking pair $(m_2, w_2)$ (because, $m_2$ prefers $w_2$ to $w_3$ and $w_2$ prefers
$m_2$ to $m_1$). In an upward traversal, we advance the partner of the woman $w_2$ in the blocking pair, $m_1$, to his next choice. The next choice for $m_1$ is $w_3$. This results in $w_3$ rejecting  $m_2$ and therefore
$m_2$ moves to his next choice $w_1$. This proposal, in turn, results in  $w_1$ rejecting $m_3$. Next, $m_3$ makes a proposal to $w_2$ and now $(m_2, w_2)$ is not a blocking
pair. The new proposal vector $(3,3,3)$ which corresponds to the matching $\{ (m_1, w_3), (m_2, w_1), (m_3, w_2) \}$ is a stable matching with all women getting their top choices.

There are two main differences between the GS algorithm and Algorithm $\alpha$.
The first difference is the simple check on the number of women that have been proposed until $G$.
We require $numw(G) \leq m$.
Clearly, if the number of women is equal to the number of men, then $numw(G)$ can never exceed $m$ and this check can be dropped.

The second difference is in the definition of $forbidden(G, i)$. In the standard GS algorithm, a man advances on his preference list
only when he is unmatched, i.e.,
the woman he has proposed to is either matched with someone more preferable or receives a proposal from a more preferable man.
Whenever the GS
algorithm reaches a man-saturating matching, it is a stable matching. 
For any arbitrary $I$ (for example, a man-saturating matching that is not stable), it is important to take
blocking pairs in consideration as part of the forbidden predicate. This difference can be summarized as follows.
\begin{itemize}
\item 
{\em GS Algorithm}: A man proposes to the next woman on his preference list if he is currently unmatched.
\item
{\em Algorithm $\alpha$}: A man $i$ proposes to the next woman on his preference list if he is currently unmatched or matched with a woman $q$ who is in a blocking pair.
\end{itemize}

Observe that if all men propose starting from their top choices, then the rule for Algorithm $\alpha$ becomes identical to
that for the GS Algorithm.

%
To prove the correctness of the algorithm $\alpha$, the following Lemma is crucial.

\begin{lemma}\label{lem:advance}
If $forbidden(G, i)$ holds, then there is no proposal vector $H$ such that $(H \geq G)$ and $(G[i] = H[i])$
and $H$ is a stable marriage.
\end{lemma}
\begin{proof}
Consider any $H$ such that $(H \geq G)$ and $(G[i] = H[i])$. We show that $H$ is not a stable marriage.
The predicate $forbidden(G, i)$ implies that there exists a man $j$ such that $\rho(G, i)$ prefers man $j$ to man $i$ and
the proposal by $j$ to $\rho(G, i)$ is in a proposal vector less than or equal to $G$. Since $G \leq H$, and $G[i]$ equals $H[i]$,
we get that $forbidden(H,i)$ also holds. Hence, $H$ is not a stable marriage from Lemma \ref{lem:forbidden}.
\end{proof}

A consequence of Lemma \ref{lem:advance} is that if $forbidden(G, i)$ holds, then it is safe to advance man $i$ to the next choice without
any danger of missing a proposal vector that is a stable marriage. We can now show the correctness of Algorithm $\alpha$.
\begin{theorem}\label{thm:alpha}
Algorithm $\alpha$ returns the least stable proposal vector $G \geq I$ in the proposal lattice whenever it exists. If there is no stable proposal vector greater than or equal to $I$, then
the algorithm returns null.
\end{theorem}
\begin{proof}
First suppose that a stable marriage exists that is greater than or equal to $I$. Since the set of stable marriages form a sublattice of the proposal lattice  \cite{knuth1997stable}, there
exists $H$, the least proposal vector that is a stable marriage and greater than or equal to $I$. Consider any $I \leq G < H$. By definition of $H$,
$G$ is not a stable marriage and there exists $i$ such that $forbidden(G, i)$ due to Lemma \ref{lem:forbidden}. From Lemma \ref{lem:advance},
the advancement along $i$ guarantees that $G \leq H$.
Hence, the algorithm will continue to advance $G$ until it is identical to $H$.

Now suppose that there is no stable marriage that is greater than or equal to $I$. In this case, the algorithm will continue to find $i$ such that
$forbidden(G, i)$ until some man runs out of choices. If we run out of choices, then from repeated application of Lemma \ref{lem:advance}, there is no stable marriage which is greater than or equal to $I$.
%
\end{proof}



The following Corollary states that the stable marriage returned by Algorithm $\alpha$ has the least distance of all stable marriages greater than $I$.
\begin{corollary}
Given any proposal vector $I$, Algorithm $\alpha$ returns the stable marriage greater than or equal to $I$ with the least distance from $I$.
\end{corollary}
\begin{proof}
Suppose that Algorithm $\alpha$ returns $G$ and $G'$ is any other stable marriage such that $I \leq G'$.
From Theorem \ref{thm:alpha}, we get that $I \leq G \leq G'$. It follows that the distance between $I$ and $G$ is less than or equal to the 
distance between $I$ and $G'$.
\end{proof}

As another application of Algorithm $\alpha$ consider a scenario where we have a stable marriage and a new man joins the system (we can assume that 
initially the number of women were more than the number of men). 
Instead of running the GS from scratch, algorithm $\alpha$ can start from the
existing proposal vector for existing men and the median choice for the new entrant. If the existing matching had certain desirable properties (e.g. fairness), then
the new stable matching found would be close to the existing matching.

%

\subsection{An Efficient Implementation of Algorithm $\alpha$}
In this section, we give an $O(m^2+w)$ implementation of Algorithm $\alpha$.
The GS algorithm maintains two data structures.
First, it maintains a list of men who are unmatched and
must advance on their preference list. We also maintain
$mList$, a list of men that are forbidden in the current proposal vector.
Second, the GS algorithm maintains an array $partner$ such that $partner[q]$ returns the partner for
woman $q$ in the proposal vector $G$. 
We maintain an array $curBest$ instead of $partner$ such that $curBest[q]$ returns the most
preferred match for woman $q$ among all proposals such that the proposal to woman $q$
by man $i$ is for a choice less than or equal to $G[i]$.
Note that if the most  preferred match for $q$ is before
$G$, then $partner[q]$ may be zero even though $curBest[q]$ is nonzero.
Such a scenario cannot happen in the GS algorithm because a man can advance on his list
only when his current matched woman gets a better proposal.
For example, consider the men's preference lists in Fig. \ref{fig:SMP2}.
Suppose the initial vector is $(2,2,1)$. In this proposal vector, $curBest[2]$ (the current best proposal to $w_2$) is from $m_2$. However, $m_2$ has moved on to 
the next woman on his list $w_3$. Hence, $m_2$ is not the partner of $w_2$.

\begin{figure}
\fbox{\begin{minipage}  {\textwidth}\sf
{\bf input}: Stable marriage instance $mpref$ and $wrank$, an initial proposal vector $I$\\
{\bf output}: the smallest stable marriage greater than or equal to $I$ (if one exists)\\
\\
$mList$: list of men initially empty; // men who need to advance\\
$G$: array[$1..m]$ of $1..w$;\\
$curBest$:array[$1..w$] of $0..m$ initially $0$; // current best proposal until $G$\\
$numw$: int initially $0$; // number of women who have received proposals until $G$\\
\\
$G$ := $I$;\\
if ($\exists i: G[i] > m$)  {\bf return}  null;  // "no stable matching exists"\\
\\
// Step 1: initialize $curBest$\\
{\bf for} $i \in [1..m]$ do\\
\h {\bf for} $k \in [1..G[i]]$ do\\
\h \h $q := mpref[i][k]$;\\
\h \h  if $(curBest[q] =0)$ then // first proposal to $q$ encountered\\
\h \h \h $curBest[q] := i$; \\
\h\h\h $numw := numw+1$; \\
\h \h \h if ($numw > m$) {\bf return}  null;  // "no stable matching exists"\\
\h \h  else if $(wrank[q][i] < wrank[q][curBest[q]])$ then\\
\h \h \h $curBest[q] := i$; \\
\\
// Step 2: initialize $mList$\\
{\bf for} $i \in [1..m]$ do\\
\h  $q := \rho(G,i)$;\\
\h if $(curBest[q] \neq i)$ then \\
\h \h append $i$ to $mList$;\\
\\
// Step 3: Advance on elements from $mList$\\
{\bf while} $(mList \neq \{\})$ \\
\h $i$ :=  first element in $mList$;\\
\h if ($G[i] < m$)) $G[i] := G[i]+1$; // try the next choice \\
\h else {\bf return}  null;  // "no stable matching exists"\\ 
\h $q := \rho(G,i)$; // woman for that choice number\\
\h  if $(curBest[q] =0)$ then\\
\h \h   $numw := numw + 1$;\\
\h\h   if ($numw > m$) then {\bf return}  null;  // "no stable matching exists"\\
\h \h $curBest[q] := i$; \\
\h \h  remove $i$ from $mList$;\\
\h else if $(wrank[q][i] < wrank[q][curBest[q]])$ then\\
\h \h if $(\rho(G,curBest[q]) = q)$ then append $curBest[q]$ to $mList$;\\
\h \h  $curBest[q] := i$; \\
\h \h  remove $i$ from $mList$;\\
{\bf endwhile};\\
{\bf return} $G$;
\end{minipage}
} 
\caption{An Implementation of Algorithm $\alpha$ with $O(m^2+w)$ complexity. \label{fig:impl}}
\end{figure}


The implementation shown in Fig. \ref{fig:impl} takes as input a stable marriage instance
specified by $mpref$ (the men's preferences) and $wrank$ (the women's ranking of men), and 
the initial proposal vector $I$. Step 1 goes over all proposals made in $G$ or before $G$ and
computes the current best proposal for every woman $q$. 
It also counts the number of unique women who have received proposals. 
If that number exceeds $m$, the algorithm returns null.

Step 2 goes over all men whose
proposals are not the current best proposals and inserts them in $mList$.
All man $i$ in $mList$ are such that they satisfy $forbidden(G, i)$.

Step 3 advances $G$ over proposals that are forbidden. It takes out
an element  $i$ from $mList$. It advances to the next choice of woman $q$.
If $q$ does not have any proposal and proposing to her does not increase $numw$ beyond $m$, 
or if $i$ is preferred over the current best proposal for $q$ so far, 
we have succeeded in removing $i$ from $mList$. If $q$ had a partner, then that man is added to $mList$.
If we run out of choices for any man $i$, then
the algorithm returns null. If the while loop terminates, we have that there is 
no $i$ such that $forbidden(G, i)$ and therefore $G$ is a stable marriage and it is returned.

The correctness of the implementation easily follows from Lemma \ref{lem:invariants} and Lemma \ref{lem:advance}.
\begin{lemma}\label{lem:invariants}
The while loop in Fig. \ref{fig:impl}  satisfies the following invariants.
\begin{enumerate}
\item
For all $q$: $curBest[q]$ is the highest ranked proposer to $q$ who has made proposal to $q$ in any proposal vector less than or equal to $G$.
\item
For all $i$: $forbidden(G, i)$ iff $curBest[\rho(G,i)] \neq i$.
\item
$i \in mList$ iff $curBest[\rho(G,i)] \neq i$.
\end{enumerate}
\end{lemma}
\begin{proof}
\begin{enumerate}
\item 
Step 1 establishes the invariant for $G$ equal to $I$.
In the while loop of Step 3, whenever $G[i]$ is incremented such that woman $q$ is proposed by $i$,
the invariant is maintained by updating $curBest[q]$.
\item
Let $q$ be equal to $\rho(G, i)$.
We first show that $curBest[q] \neq i$ implies $forbidden(G, i)$.
 If $q$ has another proposal in $G$ and $curBest[q]$ is not equal to $i$,
then man $i$ is unmatched in $G$ and therefore $forbidden(G, i)$ holds. If $q$ does not have
any other proposal in $G$, then $curBest[q] \neq i$ implies $forbidden(G, i$.

Conversely, if $forbidden(G,i)$ then either $i$ is unmatched in $G$ or her match in $G$ is part of a
blocking pair. In either case, there exists $j \neq i$ such that $q$ prefers $j$ to $i$ and
$j$ has a proposal in $G$ or before $G$.
Hence, $curBest[q] \neq i$. 
\item
Step 2 establishes the invariant at the beginning of the while loop.
Any man $i$ is removed from $mList$ only when $curBest[\rho(G, i)] = i$.
\end{enumerate}
\end{proof}

$G$ is advanced on index $i$ only when $forbidden(G,i)$. Lemma \ref{lem:invariants} 
proves that this holds iff $i \in mList$. When $mList$ is empty, none of the components of $G$ are forbidden and it corresponds to 
a stable marriage. If it is not possible to advance on some forbidden index $i$ or if more than $m$ woman have been proposed
before $G$, the algorithm returns null in accordance with Lemma \ref{lem:ubound} and Lemma \ref{lem:advance}.

Let us analyze the time complexity of Algorithm $\alpha$.
Initialization takes $O(m+w)$ time. 
Step 1 takes $O(m^2)$ time because $G[i] \leq m$. Step 2 takes $O(m)$ time to compute $mList$.
Step 3 increases $G[i]$ for some $i$ in every iteration. Each iteration can be done in $O(1)$ time and no $G[i]$ can exceed $m$.
Thus, the algorithm takes $O(m^2+w)$ time which reduces to the standard $O(n^2)$ time complexity of the GS algorithm when both $m$ and $w$ are equal to $n$.
The algorithm does not process more than $m$ choices for any man even if the number of women exceeds the number of men. This is sufficient
because there cannot exist any stable marriage that includes a choice beyond the choice number $m$ for any man.

\section{Algorithm $\beta$: Downward Traversal}

We now give the dual of Algorithm $\alpha$ that does the downward traversal in the proposal vector lattice and 
returns the greatest stable marriage less than or equal to $I$.
In the standard literature, one does not consider the dual of the GS algorithm to find the woman-optimal stable marriage.
Just by switching roles of men and women from the man-optimal GS, we get the woman-optimal GS algorithm.
We cannot employ this strategy because
we had assumed that the number of men is less than or equal to the number of women. Switching men and women
violates this assumption.
\begin{figure}
\vspace*{-0.1in}
\fbox{\begin{minipage}  {\textwidth}\sf
{\bf input}: A stable marriage instance, initial proposal vector $I$\\
{\bf output}: greatest stable marriage less than or equal to $I$ if one exists\\
\\
 The predicate $rForbidden(G, i)$ holds if there exists a woman $q$ such that 
 $q$ prefers $i$ to all her proposals in any vector less than or equal to $G$, and $i$ prefers $q$ to $\rho(G, i)$.\\
 $L$: proposal vector corresponding to man optimal stable marriage.\\
%
%
(1) for all $i$:  $G[i]:= \min(m, I[i])$;\\
(2) if $(\exists i: G[i] < L[i])$ {\bf return}  null; // no stable matching exists \\
(3) {\bf while} (there exists a man $i$ such that $rForbidden(G, i)$)\\
(4)  \h     $G[i]$:= rank of woman $q$ s.t. $q$ prefers $i$ of all proposers in any vector less than $G$\\
(5) {\bf endwhile}; \\
(6) return $G$;
 \end{minipage}
} 
  \caption{Algorithm $\beta$: An Algorithm that returns the woman-optimal marriage less than or equal to the given proposal vector $I$.\label{fig:beta}}
\end{figure}

In addition, the downward traversal of the proposal lattice gives different insights into the algorithm for finding a stable matching
even when the number of men equals the number of women.


We first give a necessary condition for a stable marriage to exist that is less than or equal to $I$.
\begin{lemma}
If $numw(I)$ (the number of unique women who have proposals in any vector less than or equal to $I$) is less than $m$,
then there cannot be any stable proposal vector less than or equal to $I$.
\end{lemma}
\begin{proof}
The claim follows because any proposal vector less than or equal to $I$ cannot be a man-saturating matching if the number of unique women is less than $m$.
\end{proof}

If $numw(I) \geq m$, there may or may not be a proposal vector that corresponds to a stable marriage depending upon the
women's preferences. 

While traversing the proposal lattice in the downward direction, we use the predicate 
$rForbidden(G, p)$ (short for reverse-Forbidden) which uses the notion of
{\em best blocking pair}.
\begin{definition}[Best blocking pair]
A blocking pair $(p,q)$ is the {\em best blocking pair} in $G$ for $q$ if for all blocking pairs $(p',q)$ in $G$, the woman $q$ prefers $p$ to $p'$.
\end{definition}
\begin{definition}[rForbidden]
A man $p$ is {\em rForbidden} in $G$ if there exists a woman $q$ such that 
$(p,q)$ is the best blocking pair in $G$ for $q$.
\end{definition}

We first show that 
\begin{lemma}\label{lem:rforbidden}
Let $G$ be any proposal vector such that $numw(G) \geq m$.
There exists a man $i$ such that $rForbidden(G, i)$ iff $G$ is not a stable marriage.
\end{lemma}
\begin{proof}
First suppose that there exists $i$ such that $rForbidden(G, i)$. Then, there exists a woman $q$ such that man $i$ and woman $q$ form a
blocking pair for $G$. Hence $G$ is not a stable marriage.

Conversely, assume that $G$ is not a stable marriage.
If $G$ is not a man-saturating matching, then there exists at least one woman $q$
who has not received any proposal in $G$ but has received it earlier because
$num(W) \geq m$. Of all such proposals to $q$ let the most favorable proposal be from man $i$. Then, $rForbidden(G, i)$ holds.
If $G$ is a man-saturating matching, but not stable, then there exists at least one blocking
pair. Therefore, there exists at least one best blocking pair.

\end{proof}

Analogous to upward traversal using forbidden predicate, we get that 
\begin{lemma}\label{lem:downward}
Assume $numw(G) \geq m$.
If $rForbidden(G, i)$ holds, then there is no stable proposal vector $H$ such that $(H \leq G)$ and $(G[i] = H[i])$.
\end{lemma}
\begin{proof}
If $rForbidden(G, i)$ holds, there exists a woman $q$ such that $q$ prefers $i$ to all men who have proposed to $q$ until $G$.
First suppose that $H$ does not have any proposal to $q$. Then, $H$ cannot be stable because $q$ is single and man $i$ prefers
$q$ to $\rho(H, i)$.
Now suppose that $H$ has a proposal to $q$. Since $H \leq G$, we know that any proposal to $q$ in $H$ is less preferable to
that by man $i$. Hence, $(i,q)$ continues to be a blocking pair in $H$.
\end{proof}

Fig. \ref{fig:beta} shows a high-level description of a downward traversal of the proposal lattice. 
At line (1), we ensure that $G[i]$ is at most $m$ because due to Lemma \ref{lem:ubound}, we know that
there cannot be any stable marriage in which any component exceeds $m$.
At line (2), we first ensure that $G$ is at least equal to $L$, the proposal vector corresponding to the man-optimal 
stable marriage.
Otherwise, there
cannot be a stable marriage vector less than or equal to $G$.
At line (3) we pick $i$ such that $rForbidden(G, i)$ holds.
This means that there exists a woman $q$ such that $(i,q)$ is a best blocking pair.
At line (4), we satisfy the pair $(i,q)$ by decreasing $G[i]$ until $\rho(G, i) = q$.
This step corresponds to a downward traversal in the proposal lattice.
At line (6),  when we exit from the while loop, we know that $G$ must be a stable
marriage on account of Lemma \ref{lem:rforbidden}.
This algorithm ensures that the match for any man can only improve.


For example of Algorithm $\beta$, consider the initial proposal vector $G=(2,2,2)$. The pair $(m_2, w_2)$ is blocking. Of all the blocking pairs in $G$ for $w_2$, 
$m_2$ is best. Even though $w_2$ prefers $m_3$ to $m_2$, the pair $(m_3, w_2)$ is not blocking because $m_3$ is at his choice $2$ in $G$ and $w_2$
corresponds to his third choice. Since $m_2$ is the best blocking pair for $w_2$, we make $m_2$ propose to $w_2$. Hence, the new proposal vector is
$(2,1,2)$. In this proposal vector, $w_3$ is unmatched and $(m_3, w_3)$ is the best blocking pair for $w_3$. The new proposal vector is $(2,1,1)$. 
Now, $w_1$ is unmatched and her best blocking pair is $(m_1, w_1)$. When $m_1$ proposes to $w_1$, we get the stable marriage proposal vector $(1, 1, 1)$.
This corresponds to the man-optimal stable marriage.

\subsection{An Efficient Implementation of Algorithm $\beta$}

We now discuss an efficient implementation of the downward traversal shown in Fig. \ref{fig:beta-impl}.
The implementation uses the following data structures. The variable $wchoice$ is the rank of the best proposal that a woman
may have in any vector less than or equal to $G$. If $wchoice[q] = k$, then the woman $q$ has a proposal from
her $k^{th}$ top choice in some vector less than or equal to $G$.
The variable $wList$ is the list of all women who may be part of some blocking pair in $G$.
When $wList$ becomes empty we know that $G$ is stable.
The variable $mrank[p][q]$ stores the rank of woman $q$ for man $p$. If $mrank[p][q]$ equals $1$, then $q$ is the top choice for man $p$.

\begin{figure}
\fbox{\begin{minipage}  {\textwidth}\sf
{\bf input}:  $mpref, wpref$: men and women preferences;\\
\h   $mrank$: men ranking of women,\\
\h $I$:  initial proposal vector;\\
{\bf output}: the greatest stable marriage less than or equal to $I$ (if one exists)\\
\\
$wList$: list of women; // women in blocking pairs\\
$wchoice$: array[$1..w]$ of $1..m$ ;\\
$G$: array[$1..m]$ of $1..w$;
\\
// Step 1: Set $G$ to $I$ ensuring that it does not exceed $m$\\
\h \h  $G[i] := \min(I[i], m)$;\\
\\
// Step 2: Trim $G$ so that proposals are only to women in man-optimal marriage $L$\\
\h  $L := \alpha((1,1,\ldots,1))$;\\
\h  if $(\exists i: G[i] < L[i])$ {\bf return}  null;  // "no stable matching less than or equal to $G$ exists"\\
\h $W'$ := set of women matched in $L$;\\
\h {\bf for} $i \in [1..m]$ do\\
\h \h $k := L[i]$;\\
\h \h {\bf while} $(k < G[i]) \wedge (mpref[i][k+1] \in W')$\\
\h \h \h  $k := k+1$;\\
\h \h   $G[i] := k$;\\
\\
// Step 3: Satisfy all potential blocking pairs \\
\h $wList := W'$;\\
\h {\bf for} $j \in W': wchoice[j] := 1$;\\
\h {\bf while} $(wList \neq \{\})$ \\
\h \h $q$ :=  first element in $wList$;\\
\h \h $p := wpref[q][wchoice[q]]$;\\
\h \h  if   $(mrank[p][q] = G[p])$ then //q is assigned to her best choice \\
\h \h \h   remove $q$ from $wList$;\\
\h \h else if $mrank[p][q] > G[p]$ then// this man has not proposed to $q$ until $G$\\
\h\h\h$wchoice[q]++$ ;\\
\h\h  else //  $(mrank[p][q] < G[p])$ \\
\h \h \h  $wList : = wList \cup \rho(G, p)$; // add the current woman partner for $p$ to $wList$\\
\h \h \h $G[p] := mrank[p][q]$; // satisfy the blocking pair $(p,q)$\\
\h \h \h   remove $q$ from $wList$;\\
\h {\bf endwhile};\\
\h  {\bf return} $G$;\\
\end{minipage}
} 
\caption{An Implementation of Algorithm $\beta$ with $O(m^2+w)$ complexity. \label{fig:beta-impl}}
\end{figure}

In the first step, we ensure that no component of $G$ exceeds $m$.
Since there are $m$ men and $G[i] \leq m$, there are $O(m^2)$ proposals up to $G$.
These proposals may be to different women. In step 2, we further trim $G$ to guarantee that 
we never have to handle more than $m$ women. Recall that our goal is to get an algorithm
whose complexity is at most linear in the number of women.
To that end, we first claim the following:
\begin{theorem}\label{thm:trimG}
Let $W'$ be the set of women that are part of the man-optimal marriage. \\
(a) Any stable marriage can only include women from $W'$.\\
(b) Any matching $G$ such that a man $i$ prefers a woman outside of $W'$ to $\rho(G, i)$ is not stable.
\end{theorem}
\begin{proof}
(a) Consider any $G$ that is a man-saturating matching with a woman outside of $W'$. Suppose for contradiction $G$
is a stable matching. Let $G_{o}$ be the man optimal stable matching. Then $G_o \leq G$.
However, this means that the number of unique women proposed until $G$ exceeds $m$. Hence, $G$
or any proposal vector greater than $G$ cannot be a stable marriage from Lemma \ref{lem:ubound}.\\
(b) Consider any $G$ that is a man-saturating matching such that some man $i$ prefers  a woman outside of $W$ to $\rho(G, i)$. 
Suppose $G$ is stable. From part (a), we conclude that $G$ includes proposal to $m$ women in $W'$.
Furthermore, if man $i$ prefers a woman outside of $W'$, we get that the number of unique women proposed until $G$ exceeds $m$.
Hence, from Lemma \ref{lem:ubound}, $G$ cannot be stable.
\end{proof}

We exploit Theorem \ref{thm:trimG} as follows. At step 2, we find the man-optimal stable marriage $L$. 
If there is any component $i$ of $G$ such that $G[i]$ is less than $L[i]$, then clearly there cannot be a stable marriage
before $G$. Now based on $L$, we can determine $W'$, the set of women that can be part of any stable marriage.
We now decrease $G[i]$ for each $i$ such that proposals are made until $G[i]$ are only to women in $W'$.
This step would not be required if the number of men and women were equal.


In the third step, we ensure that there is no man $i$ such that $rForbidden(G, i)$. Instead of maintaining forbidden men, it is easier to 
maintain $wList$, the list of all women in $W'$ who may be part of blocking pairs in $G$. We initialize $wList$ to all women in $W'$.
Whenever $wList$ becomes empty, there are no blocking pairs.
To satisfy a blocking pair, we first remove a woman $q$ from $wList$. 
We determine $p$, the top choice of that woman $q$ in a vector less than or equal to $G$.
There are three possibilities:
(1) If $p$ is matched with $q$ in $G$, then we are done with processing of $q$ and $q$ is deleted from $wList$.
(2) If $p$ has not made proposal to $q$ in any vector before $G$, then the woman $q$ must move on to her next choice.
(3) If $p$ had a proposal to $q$ before $G$ but not in $G$, then
we have that $(p,q)$ is a blocking pair. Since we are exploring choices in the order of $wchoice$, $(p,q)$ is the best blocking pair for $q$.
We satisfy this blocking pair by moving $G[p]$ to the choice given by the
woman $q$. This step corresponds to a downward traversal in the proposal lattice and the assignment for man $p$ improves
with this step. Since $q$ has her best proposal, she is deleted from $wList$.
The current partner of $p$ is added to $wList$, if not already on the list.

In the algorithm $\beta$, $wchoice$ for any woman can only increase. As $wchoice$ increases, women are assigned less preferable choices.
Similarly, the vector $G$ which corresponds to choices by men only decreases with the algorithm. Hence, the assignment to men only improves
with the execution of the algorithm.

We now have the following result that shows correctness of the implementation of Algorithm $\beta$.

\begin{theorem}\label{thm:beta-impl}
Given any proposal vector $I$, Algorithm $\beta$ in Fig. \ref{fig:beta-impl} returns the greatest stable proposal vector less than or equal to $I$.
\end{theorem}
To show
the correctness of the implementation, we first show the following Lemma.

\begin{lemma}\label{lem:beta-impl}
The {\em while} loop in Step 3 of Fig. \ref{fig:beta-impl} has the following invariants.
\begin{enumerate}
\item
For any woman $q$, $wchoice$ is always less than or equal to the top choice she has before $G$.
\item
Any woman $q$ such that $(p,q)$ is a best blocking pair, is included in $wList$.
\item 
Any proposal vector $H$ such that $H$ is less than or equal to $I$ and $H[i] > G[i]$ for some $i$
is not a stable marriage.
\end{enumerate}
\end{lemma}
\begin{proof}
\begin{enumerate}
\item
It is initially true because $wchoice[q]$ is initialized to $1$ for all $q \in W'$.
It is incremented only when $mrank[p][q] > G[p]$, i.e., the man corresponding to $wchoice[q]$
has proposal to $q$ after $G[p]$.
\item 
It is true initially because every woman in included in $wList$ before the while loop.
A woman is removed from $wList$ only when $G[p]$ equals $mrank[p][q]$ where $p$
is her top choice before $G$, i.e., $q$ is matched to her top choice of all proposals made at or before $G$.
\item
The invariant holds initially because Step 1 and 2 decrease $G[i]$ from $I[i]$ only when
it either is beyond $m^{th}$ choice (step 1), or to a woman outside of $W'$ (step 2).
Now suppose that the invariant holds at the beginning of the while loop.
We decrease $G[p]$ to $mrank[p][q]$ only when there exists a woman $q$ such that 
$mrank[p][q]$ is less than $G[p]$. Consider any proposal vector such that $H[p]$ is greater than
$mrank[p][q]$. By part 1 of this Lemma, $q$ prefers $p$ to any assignment it may have in $H$.
Also, man $p$ prefers $q$ to any assignment it may have in $H$ (because $mrank[p][q] < H[p]$).
Hence, $(p,q)$ is a blocking pair in $H$.
\end{enumerate}
\end{proof}

We now give the proof of Theorem \ref{thm:beta-impl}.
\begin{proof}
We initialize $G$ to $I$. It is sufficient to show that any component $p$ of $G$ is decreased
iff $rForbidden(G,p)$. First, suppose that $rForbidden(G,p)$ holds. This implies that
there exists a woman $q$ such that $(p,q)$ is a best blocking pair. From Lemma \ref{lem:beta-impl} part 2, $q$ is included in $wList$. Whenever $q$ is the first element in $wList$, 
$mrank[p][q]$ is less than $G[p]$ and $G[p]$ is decreased.
Conversely, $G[p]$ is decreases when $mrank[p][q] < G[p]$. This implies $rForbidden(G, p)$.
\end{proof}

The number of women who have been proposed until $G$ at the beginning of step 2,
 is at most $O(m)$ irrespective of the total number of women in the system.
The {\em while} loop removes a woman from $wList$, increases $wchoice$ for at least one woman,
or decreases $G[p]$ for some man $p$ in every iteration. Each of these actions can be done
at most $O(m^2)$ number of times over all iterations giving us the
time complexity of Algorithm $\beta$ as $O(m^2+w)$.

\section{Algorithm $\gamma$: Path to Stability}  
We now present an algorithm that gives a path from any proposal vector to a stable marriage vector.
Note that 
depending on the initial proposal vector, both Algorithms $\alpha$ and $\beta$ may return null. For example, when the number of men is equal to the number of women
and the initial vector is greater than or incomparable to the woman-optimal vector, then the algorithm $\alpha$ will return null.
Similarly, if the initial vector is less than or incomparable to the man-optimal vector, then the algorithm $\beta$ returns null.
If the initial vector is incomparable to both the man-optimal and the woman-optimal proposal vectors, then both algorithms $\alpha$ and
$\beta$ will return null.
In the RVV setting,
we need to combine
a downwards traversal with an upwards traversal to go from an arbitrary proposal vector to a stable matching.
There are two choices for combining these traversals --- a downward traversal followed by an upwards traversal, or vice-versa.
We will use the former approach.
The RVV algorithm introduces men and women incrementally and does multiple upward and downward traversals.
\vspace*{-0.1in}
\begin{figure}
\fbox{\begin{minipage}  {\textwidth}\sf

{\bf input}: A stable marriage instance, initial proposal vector $I$\\
{\bf output}: a stable marriage $M$\\
%
$G := I$;

// Downward traversal \\
$K := max(G, U)$; //Compute $U$ using Algorithm $\beta$ with the initial vector as $[m,m,\ldots,m]$;\\
{\bf while} there exists a  man $i$ such that $rForbidden(K, i)$\\
\h \h $K[i] := K[i] - 1;$ $G[i] := G[i] - 1;$\\
{\bf endwhile};\\
//Upward traversal\\
{\bf while} there exists a  man $i$ such that $forbidden(G, i)$\\
\h $G[i] := G[i] + 1;$\\
 {\bf endwhile};\\
 return $G$;
\end{minipage}
} 
\caption{ Algorithm $\gamma$ with $O(m^2+w)$ complexity. \label{fig:gamma}}
\vspace*{-0.1in}
\end{figure}

The Algorithm $\gamma$ is shown in Fig. \ref{fig:gamma}.
Given any arbitrary initial vector $I$, we first do a downward traversal to get to a proposal vector
that is less than or equal to $U$, the largest possible stable marriage. If the 
initial vector is at most $U$, then this step is not necessary.
$U$ can be computed using Algorithm $\beta$ by using a downward traversal starting from the vector
$[m,m,\ldots,m]$. Our goal is to find blocking pairs in $G$ such that by satisfying them we
get to a proposal vector $G \leq U$. In contrast to algorithms in literature, we pick blocking pairs to satisfy carefully.
Specifically, during the downward traversal, we satisfy only those men whose component in the proposal vector
is beyond $U$.
To find a sequence from $I$ to $G$ such that $G \leq U$, we first compute a vector
$K$ as $\max(U, G)$. 
We now invoke a downward traversal on $K$ using $rForbidden$ function of algorithm $\beta$. Since Algorithm $\beta$ returns the greatest stable marriage less than the initial 
proposal vector (in our case $K$), it finds as blocking pair only those men $i$
such that $K[i] > U[i]$. By definition of $rForbidden$ any $j$ such that $K[j]$ equals $U[j]$ can not satisfy $rForbidden(K, j)$ because $U$ is a stable marriage. 

\begin{lemma}
Let $K = max(G, U)$. Then, for any $i$,
$rForbidden(K, i)$ implies $rForbidden(G,i)$.
\end{lemma}
\begin{proof}
Suppose $i$ is not $rForbidden$ in $G$. This means that there exists a stable marriage $H$ less than or equal to $G$ such that $H[i] = G[i]$. Since $G$ is less than or equal to $K$, 
we get that $H$ is a stable marriage less than or equal to $K$. However, this implies that
$i$ is not $rForbidden$ in $K$.
\end{proof}
Since $i$ is $rForbidden$ in $G$ it is safe to decrement $G[i]$ in search for a stable marriage. By repeating this process, we generate a sequence of proposal vectors that makes $G$ less than or equal
to $U$. Note that consecutive proposal vectors generated in this phase differ in
the proposals by at most one man. 
The downward traversal step can be viewed as invocation of Algorithm $\beta$ on $K$ such that whenever $K$ is updated, $G$ is updated as well. This downward traversal can be done in
$O(m^2+w)$ time. At the end of this step $G \leq U$, and we can start the second phase of the algorithm.

In the second phase, we do an upward traversal in which women improve their match.
We use the function $\alpha$ to find the least stable marriage that is greater than or equal to $G$.
In this phase, we satisfy blocking pairs by improving the match of women. Since the input to algorithm $\alpha$
is less than or equal to $U$, we are guaranteed to get a stable marriage at the end.

Hence, we have the following result.
\begin{theorem}
Given any initial proposal vector $I$, there exists a sequence of 
proposal vectors $G_0, G_1, \ldots, G_t$ such that
$G_0$ is equal to $I$, $G_t$ corresponds to a stable matching and for
each  $k$ ($1 \leq k \leq t$), $G_k$ is obtained from $G_{k-1}$ by either increasing the choice number for one man (thereby worsening his match) or
decreasing the choice number for one man (thereby improving his match). The value of $t$ is at most $2m^2$ where $m$ is the number of men.
\end{theorem}

This sequence can be obtained using algorithm $\gamma$ that takes $O(m^2+w)$ computation time
given all the data structures (preference lists and rankings) in memory.

\begin{figure}
\fbox{\begin{minipage}  {\textwidth}\sf

{\bf input}: A stable marriage instance, initial proposal vector $I$\\
{\bf output}: a stable marriage $M$\\
\h compute $U$ using Algorithm $\beta$ with the initial vector as $[m,m,\ldots,m]$;\\
\h $G := I$;
\h    $K := max(G, U)$;\\
\h {\bf while} $\exists$ a  man $i$ such that $(i,q)$ is a best blocking pair in $K$\\
\h \h    set $K[i]$ and $G[i]$ to the choice corresponding to woman $q$;\\
\h \h    generate a matching that satisfies the blocking pair $(i,q)$;\\
\h {\bf endwhile};\\
\\
\h {\bf while} $\exists$ a  man $i$ s.t. $i$ unmatched or $(i,q)$ is a blocking pair in $G$\\
\h \h set $G[i]$ to the next woman $q$ who would accept proposal from $i$;\\
\h \h    generate a matching that moves the man $i$ from the current partner to $q$;\\
\h {\bf endwhile};\\
\h return $G$;
\end{minipage}
} 
\caption{ Algorithm $\gamma$ that generates a sequence of matchings. \label{fig:gamma2}}
\end{figure}

Since the RVV Algorithm generates a sequence of matchings instead of proposal vectors, we show how to generate a sequence of matchings explicitly instead of proposal vectors in Fig. \ref{fig:gamma2}. The downward traversal is performed by using the best blocking pairs in $K$.
The matching is generated from the proposal vector $G$ as defined in Section \ref{sec:proposal-vector-lattice}.
Observe that these matchings may not be men-saturating and therefore some men and women may be unmatched.
The upward traversal is performed by matching those men who are either unmatched or
matched to a woman in a blocking pair. Clearly, the length of the sequence of these
matchings is at most $O(m^2)$.

\section{Algorithm $\delta$: Stable Matching at the Shortest Distance}
In this section, we give an algorithm, called Algorithm $\delta$, that finds the stable matching with the least distance of all stable matchings relative to the initial proposal vector. 

Given an arbitrary proposal vector (not necessarily a matching) $I$, we want to find a stable matching $M$ such that the distance between the proposal vector $I$ and the stable matching $M$ is minimized over all stable matchings, $\mathcal{M}$. The distance we consider here is $L_1$ distance, a.k.a Manhattan distance between two vectors. We denote the distance as $dist(I, M) = \|I - M\|_1$. The problem hence can be rephrased as: find the
marriage $M \in \mathcal{M}$ that minimizes $dist(I,M)$.

It is well-known that the convex hull of stable matchings of an arbitrary bipartite preference system can be described by a linear system \cite{rothblum1992characterization} as follows: 
\begin{subequations}
\label{eq:linear}
\begin{align}
    \sum_{j \in [w]} x_{i,j} \leq 1 \hspace{10pt} &\forall i \in [m] \\
    \sum_{i \in [m]} x_{i,j} \leq 1 \hspace{10pt} &\forall j \in [w] \\
    \sum_{i' \in [m]; i' >_j i} x_{i',j} + \sum_{j' \in [w]; j' >_i j} x_{i,j'} + x_{i,j} \geq 1 \hspace{10pt} &\forall (i,j) \in [m] \times [w]
\end{align}
\end{subequations}

Here, we define that for each man or woman $i$, $p>_iq$ denotes that $i$ prefers $p$ over $q$ in his/her preference list. Rothblum \cite{rothblum1992characterization} proved that the linear system above is integral, i.e. every basic feasible solution of Equation \ref{eq:linear} is integral. Suppose that every possible marriage $(i,j)$ has a cost $c(i,j)$, we can find a minimum-cost stable matching in polynomial time by solving the LP above. 

Now we show that our problem of minimizing the distance between an initial proposal vector and any stable matching can be translated into a minimum-cost stable matching problem with a carefully designed cost function.

For each pair $(i,j)$, we assign the cost $c(i,j) = |I[i] - mrank[i][j]|$. Hence, for each stable matching $M$, we have:

    $dist(I, M) = \sum_{i \in [m]} |I[i] - M[i]|
               = \sum_{(i,j) \in [m] \times [w]} c(i,j) \cdot \mathbbm{1}_{\rho(M,i) = j}$

Hence, we can rewrite our problem as:

\begin{mini!}
    {}{\sum_{(i,j) \in [m] \times [w]} c_{i,j} \cdot x_{i,j}}{}{}{}
\addConstraint{\sum_{j \in [w]} x_{i,j} \leq 1}{\hspace{10pt}\forall i \in [m]}
\addConstraint{\sum_{i \in [m]} x_{i,j} \leq 1}{\hspace{10pt}\forall j \in [w]}
\addConstraint{\sum_{i' \in [m]; i' >_j i} x_{i',j} + \sum_{j' \in [w]; j' >_i j} x_{i,j'} + x_{i,j} \geq 1}{\hspace{10pt} \forall (i,j) \in [m] \times [w]}
\end{mini!}

Solving the above LP gives us a stable matching that is nearest to the initial proposal vector.
This LP has $O(n^2)$ variables and constraints where $n$ is the total number of men and women.
However, we note that
the minimum-cost stable matching problem can be reduced to the minimum-cost closed subset of a poset due to the rotation poset structure of stable matching problem. See \cite{gusfield1989stable} for more details of rotation poset. Feder \cite{feder1994network} has shown that the minimum-cost stable matching problem in a bipartite preference system can be solved in $O(n^3)$ time if $max(c_{i,j}) = O(n)$ .

We summarize the preceding discussion as the following theorem.
\begin{theorem}
Given an arbitrary proposal vector $I$, we can find a stable matching $M$ that minimizes the distance $dist(I,M)$ over all stable matchings in $O(n^3)$ time where $n$ is the total number of men and women.
\end{theorem}

\vspace*{-0.1in}
\section{Conclusions and Future Work}
We have proposed algorithms to find a stable matching starting from a given initial matching
or an initial proposal vector. 
Algorithm $\alpha$ ($\beta$) returns the stable matching with least distance 
of all stable matching that are better than the initial matching from the women's perspective (the man's perspective) in $O(n^2)$ time. 
Algorithm $\delta$ returns the stable matching with least distance from the initial proposal vector
in $O(n^3)$ time. The following problem is open.
Is there an efficient $O(n^2)$ algorithm that returns the {\em closest} stable proposal vector given
any initial proposal vector $I$?
\begin{small}
\bibliographystyle{plainurl}
\bibliography{fmaster}
\end{small}






\end{document}